\documentclass[10pt, conference,onecolumn]{IEEEtran}  

\usepackage{times}
\usepackage{hyperref} 
\usepackage{amsmath}
\usepackage{amssymb}
\usepackage{amsthm}
\usepackage{bbm}
\usepackage{dsfont}

\usepackage{epstopdf} 
\usepackage{graphicx}
\usepackage{url}

\usepackage[normalem]{ulem}

\usepackage{color}

\DeclareMathOperator*{\argmin}{arg\,min}

\newtheorem{theorem}{Theorem}
\newtheorem{lemma}{Lemma}

\theoremstyle{definition}
\newtheorem{definition}{Definition}
\theoremstyle{remark}
\newtheorem{remark}{Remark}

\newcommand{\mynewline}{\mbox{}\\}

\newcommand{\ones}{{\textbf 1}}

\newcommand{\BRA}[1]{\left( #1 \right)}

\newcommand{\BRAs}[1]{\left\{ #1 \right \}}

\newcommand{\cA}{{\mathcal A}}
\newcommand{\cX}{{\mathcal X}}
\newcommand{\cY}{{\mathcal Y}}

\newcommand{\cP}{{\mathcal P}}

\newcommand{\bx}{\textbf{x}}
\newcommand{\by}{\textbf{y}}
\newcommand{\bz}{\textbf{z}}

\newcommand{\ie}{{\emph{i.e.}}}
\newcommand{\eg}{{\emph{e.g.}}}

\newcommand{\etal}{{\emph{et al.}}}

\newcommand{\Ind}[1]{ \mathds{1}_{\BRAs{#1}} }

\newcommand{\MIN}[1]{ \smash{\displaystyle\min_{#1}} }

\newif\ifFullProofs
\FullProofstrue 

\newif\ifTwoCols
\TwoColsfalse 

\newcommand{\ColVec}[2]{\BRA{ \begin{array}{c}  #1 \\#2 \\ \end{array}    }}
\newcommand{\TripleColVec}[3]{\BRA{ \begin{array}{c}  #1 \\#2 \\#3\\ \end{array}    }}

\newcommand{\LPMat}{\BRA{ \begin{array}{ccc} A & 0 & C\\ I & I & 0\\ \end{array}}}
\newcommand{\Wyx} {W_{Y|X}\BRA{\by|\bx}}
\newcommand{\Tx} {T_{\bx}}
\newcommand{\Ty} {T_{\by}}
\newcommand{\bzy} {\bz_{\by}}

\begin{document}

\title{On the calculation of the minimax-converse of the channel coding problem}

\author{Nir~Elkayam ~~~~~~~Meir~Feder \\
        Department of Electrical Engineering - Systems\\
        Tel-Aviv University, Israel \\
        Email: nirelkayam@post.tau.ac.il, meir@eng.tau.ac.il}

\maketitle

\subsection*{\centering Abstract}
\textit{
A minimax-converse has been suggested for the general channel coding problem \cite{polyanskiy2010channel}. This converse comes in two flavors. The first flavor is generally used for the analysis of the coding problem with non-vanishing error probability and provides an upper bound on the rate given the error probability. The second flavor fixes the rate and provides a lower bound on the error probability. Both converses are given as a min-max optimization problem of an appropriate binary hypothesis testing problem. The properties of the first converse were studies in \cite{polyanskiy2013saddle} and a saddle point was proved. In this paper we study the properties of the second form and prove that it also admits a saddle point. Moreover, an algorithm for the computation of the saddle point, and hence the bound, is developed. In the DMC case, the algorithm runs in a polynomial time.}

\section{Introduction}

Achievable and Converse bounds were derived in \cite{ElkayamITW2015} for the problem of point to point (P2P) channel coding by using the standard \textbf{random coding} argument. The setting considered a general channel and a general (possibly mismatched) decoding metric. Both achievable and converse results were given in terms of a function $F(R)$, which is the cumulative distribution function (CDF) of the pairwise error probability. When the decoding metric is matched to the channel (which is the focus of this paper), the converse bound reduces to the \textbf{minimax converse}, proposed in \cite{polyanskiy2010channel}.

Consider an abstract channel coding problem; that is a random transformation defined by a pair of measurable spaces of inputs $\cX$ and outputs $\cY$ and a conditional probability measure $W_{Y|X}: \cX \mapsto \cY$. Let $M$ be a positive integer. A flavor of the minimax converse is a lower bound on the error probability of any code with $M=2^R$ codewords. The proof of the minimax converse relies on a reduction from the channel coding problem to the binary hypothesis testing problem. The bound is given in terms of $\beta_{\alpha}\BRA{P,Q}$, which is the power of the test (i.e. type II error probability) at a significance level $1-\alpha$ (i.e., type I error probability), to discriminate between probability measures $P$ and $Q$.


Specifically, the minimax converse comes in the following two flavors:
  \begin{equation}\label{MetaConverse:Intro:1}
    \epsilon \geq \inf_{Q_X}\sup_{Q_Y} \beta_{1-\frac{1}{M}}\BRA{Q_X \times Q_Y, Q_X W_{Y|X}}
  \end{equation}
  \begin{equation}\label{MetaConverse:Intro:2}
    \frac{1}{M} \geq \inf_{Q_X}\sup_{Q_Y} \beta_{1-\epsilon}\BRA{Q_X W_{Y|X}, Q_X \times Q_Y}.
  \end{equation}
where $Q_X W_{Y|X}$ and $Q_X \times Q_Y$ are the joint distributions on $\cX\times\cY$ defined by\footnote{throughout the paper, we assume that the alphabets $\cX$ and $\cY$ are finite or countably infinite.}:
\begin{align*}
  \BRA{Q_X W_{Y|X}}(\bx,\by)&=Q_X(\bx) W_{Y|X}(\by|\bx) \\
  \BRA{Q_X \times Q_Y}(\bx,\by)&=Q_X(\bx) Q_Y(\by)
\end{align*}
The first form \eqref{MetaConverse:Intro:1} gives a lower bound on the error probability of any code given that the number of codewords is $M$. The second form \eqref{MetaConverse:Intro:2} gives an upper bound on the number of codewords $M$ given that the error probability is $\epsilon$. Both bounds are given as a $\inf-\sup$ optimization problem on the set of \emph{input distributions} $Q_X$ and \emph{output distributions} $Q_Y$.

The functional properties of $\beta_{1-\epsilon}\BRA{Q_X W_{Y|X}, Q_X \times Q_Y}$, as a function of $Q_X$ and $Q_Y$ (\ie, the objective function in \eqref{MetaConverse:Intro:2}) were investigated in \cite{polyanskiy2013saddle}. In particular, the function is convex-concave and the existence of a \textbf{saddle point} was proved under general conditions. The focus of this paper is on the form \eqref{MetaConverse:Intro:1}, as this form has been used in \cite{ElkayamITW2015} for the converse and achievable results there.

Specifically, our goal in this paper is to develop tools to evaluate the optimization problem \eqref{MetaConverse:Intro:1}, and the distributions $Q_X$ and $Q_Y$ that attain it. In particular, by calculating the optimal distribution $Q_X$ in \eqref{MetaConverse:Intro:1} for a given $R=\log M$, we obtain both a converse bound and a ``good'' distribution for random coding at rate $R$, whose performance are close up to a factor to the converse result, see \cite[Theorem 4]{ElkayamITW2015} for the exact statement.

The paper is structured as
follows:
\begin{itemize}
  \item In section \ref{Sec:BinaryHypLemma} we derive a general variational formula for the functional $\beta_{\alpha}$. The formula is interesting by its own right (see further \cite{ElkayamBinaryTest}), but in this paper we are interested only in its usage for analyzing the minimax converse.
  \item In section \ref{Sec:ApplicationToMetaConverse} we apply the variational formula on the functional: $$ \beta_{1-e^{-R}}\BRA{Q_X \times Q_Y, Q_X W_{Y|X}}.$$
      This gives us a hint for defining a new functional $\gamma$ with a larger domain than $\beta$. This new functional is convex-concave, thus has a saddle point, which in turn implies a saddle point of \eqref{MetaConverse:Intro:1}. Moreover, necessary and sufficient conditions for the saddle point are proved.
  \item In section \ref{Sec:HighLevelDescription} we provide a high level description of an algorithm for computing the saddle point of $\gamma$. Following that we provide in section \ref{Sec:AlgorithmDetail} a more detailed description of the algorithm, showing how it builds a sequence of input distributions $Q_X^{(k)}$ using linear programs designed to reduce the score $ \sup_{Q_Y}\beta_{1-e^{-R}}\BRA{Q_X^{(k)} \times Q_Y, Q_X^{(k)} W_{Y|X}}$.
      \end{itemize}

  \ifFullProofs
In the appendix \ref{App:DMC} we describe the modification needed for the calculation of the minimax-converse for Discrete Memoryless Channels (DMC) where symmetries can be used to significantly reduce the computational burden into a polynomial time algorithm (as a function of the block length) for a fixed (small) $|\cX|,|\cY|$ input and output alphabet.
  \fi

\section{General Binary Hypothesis testing}\label{Sec:BinaryHypLemma}

Recall some general (and standard) definitions about the optimal performance of a binary hypothesis testing between two probability measures $P$ and $Q$ over a set $W$:
\begin{equation}\label{binary_hypothsis:beta}
  \beta_{\alpha}\BRA{P,Q} = \MIN{\substack{P_{Z|W} :\\ \sum_{w\in W}P(w)P_{Z|W}(1|w) \geq \alpha} } \sum_{w\in W}Q(w)P_{Z|W}(1|w),
\end{equation}
where $P_{Z|W}:W \rightarrow \BRAs{0,1}$ is any randomized test. The minimum is guaranteed to be achieved by the
Neyman--Pearson lemma. Thus, $\beta_{\alpha}\BRA{P,Q}$ gives the minimum probability of error under hypothesis $Q$ if the probability of error under hypothesis $P$ is not larger than $1-\alpha$. $\beta$ is the \textbf{power} of the test at \textbf{significance level} $1-\alpha$.


\begin{lemma}\label{Lemma:BinaryHyp}
The following variational formula holds:
\begin{equation}\label{Formula:Beta:Sup}
  \beta_{\alpha}\BRA{P,Q} = \max_{\lambda}\BRA{\sum_{w\in W} \min\BRA{Q(w),\lambda P(w)} - \lambda\BRA{1-\alpha}}.
\end{equation}
Moreover,
\begin{equation}\label{Formula:Beta}
  \beta_{\alpha}\BRA{P,Q} = \sum_{w\in W} \min\BRA{Q(w),\lambda P(w)} - \lambda\BRA{1-\alpha}
\end{equation}
If and only if:
\begin{equation}\label{Formula:OptimalLambda}
  P\BRAs{w:\frac{Q(w)}{P(w)} < \lambda} \leq \alpha \leq P\BRAs{w:\frac{Q(w)}{P(w)} \leq \lambda}
\end{equation}
\end{lemma}

\ifFullProofs
The proof appears in Appendix \ref{App:BinaryHypLemma}.
\fi


\section{Analysis of the minimax-converse} \label{Sec:ApplicationToMetaConverse}
\subsection{General definitions}
Consider an abstract channel coding problem; that is, a random transformation defined by a pair of measurable spaces
of inputs $\cX$ and outputs $\cY$ and a conditional probability measure $W_{Y|X}: \cX \mapsto \cY$. The notation $\cP \BRA{\cA}$ stands for the set of all probability distributions on $\cA$. Throughout this paper we assume that $|\cX| < \infty, |\cY| < \infty$. We use $\max$ and $\min$ instead of $\sup$ and $\inf$ as we generally deal with convex/concave optimization problems over compact spaces and the $\sup/\inf$ is generally attained by some element. For a distribution $Q_X\in \cP \BRA{\cX}$ and $Q_Y\in\cP \BRA{\cY}$, denote by $Q_X W_{Y|X}$ the joint distribution on $\cX\times\cY$ where $\BRA{Q_X W_{Y|X}}(\bx,\by) = Q_X(\bx)W_{Y|X}(\by|\bx)$ and $\BRA{Q_X \times Q_Y}(\bx,\by) = Q_X(\bx)Q_Y(\by)$.

\subsection{The minimax-converse}

As noted above, Polyanskiy \etal\ \cite{polyanskiy2010channel} proved the following general converse result for the average error probability that come in two flavors:
  For any code with $M$ equiprobable codewords:
  \begin{equation}\label{MetaConverse:1}
    \epsilon \geq \inf_{Q_X}\sup_{Q_Y} \beta_{1-\frac{1}{M}}\BRA{Q_X \times Q_Y, Q_X W_{Y|X}}
  \end{equation}
  \begin{equation}\label{MetaConverse:2}
    \frac{1}{M} \geq \inf_{Q_X}\sup_{Q_Y} \beta_{1-\epsilon}\BRA{Q_X W_{Y|X}, Q_X \times Q_Y}.
  \end{equation}
  where $\epsilon$ is the average error probability.
Eq. \eqref{MetaConverse:1} gives a lower bound on the error probability in terms of the rate while the second flavor, \eqref{MetaConverse:2}, gives an upper bound on the rate in terms of the error probability. Furthermore, using equation \eqref{MetaConverse:2} and instantiating $Q_Y$, it was shown in \cite{polyanskiy2010channel} that most other known converses of the channel coding problem can be derived from this converse. 
In \cite{polyanskiy2013saddle}, the functional properties of the minimax-converse \eqref{MetaConverse:2} have been further investigated. In particular, its convexity w.r.t $Q_X$ and concavity w.r.t $Q_Y$ were shown.

In this paper our focus is on the form \eqref{MetaConverse:1} as this form has been used in \cite{ElkayamITW2015} for the achievable and converse parts. The convexity of \eqref{MetaConverse:1} in $Q_X$ follows from \cite[Theorem 6]{polyanskiy2013saddle}; however, the functional is not concave with respect to $Q_Y$ in general. Applying Lemma \ref{Lemma:BinaryHyp} to this case gives the following formula:

\ifTwoCols
\else
  \begin{align*}
    \beta_{1-e^{-R}}\BRA{Q_X Q_Y, Q_X W_{Y|X}} = \max_{\lambda}\BRA{ \sum_{\bx,\by} Q_X(\bx)\min\BRA{W_{Y|X}(\by|\bx),\lambda Q_Y(\by)} - e^{-R}\lambda }
  \end{align*}
\fi

The convexity of $\beta_{1-e^{-R}}\BRA{Q_X Q_Y, Q_X W_{Y|X}}$ with respect to $Q_X$ then follows easily since it is the $\max$ of the convex (affine) function of $Q_X$. Unfortunately, $\beta$ is not concave in $Q_Y$. Yet, in order to analyze the minimax converse, we define a new function $\gamma$ over a larger domain, which (as shown below) is convex-concave:

\begin{definition}
For any distribution $Q_X \in \cP \BRA{\cX}$ and $\bz=\BRA{\bzy} \in [0,1]^{|\cY|}=\BRAs{\BRA{\bzy}\in \mathbb{R}^{|\cY|} : 0 \leq \bzy \leq 1}$\footnote{Throughout this paper $\bz$ will stand for a vector, indexed by the elements $\cY$, \ie, the component of $\bz$ are $\bz_y$.}:

\ifTwoCols
\else
\begin{equation}\label{Def:Gamma}
  \gamma_{1-e^{-R}}(Q_X,\bz,\Wyx) = \sum_{\bx,\by} Q_X(\bx)\min\BRA{W_{Y|X}(\by|\bx),\bzy} - e^{-R}\sum_{\by} \bzy
\end{equation}
\fi

Since throughout this paper $\Wyx$ and $R$ are held fixed, we will abbreviate and write $\gamma(Q_X,\bz)$ instead of $\gamma_{1-e^{-R}}(Q_X,\bz,\Wyx)$.
\end{definition}

Some properties of $\gamma(Q_X,\bz)$ are summarized in the following theorem. In particular, the functional admits a saddle point.
\begin{theorem}\label{Theorem:Gamma:Propeties}\mynewline
$\gamma(Q_X,\bz)$ is convex in $Q_X$, concave in $\bz$ and admits a saddle point $\BRA{Q_X^*, \bz^*}$, \ie
\begin{equation}\label{SaddlePoint:Value}
  \gamma(Q_X^*,\bz) \leq \gamma(Q_X^*,\bz^*) \leq \gamma(Q_X,\bz^*)
\end{equation}
for all $Q_X$, $\bz$. In particular:
\begin{equation}\label{SaddlePoint:Minimax}
  \epsilon = \min_{Q_X} \max_{\bz} \gamma(Q_X,\bz) = \max_{\bz} \min_{Q_X} \gamma(Q_X,\bz)
\end{equation}
Moreover, for $\bx$ such that $Q_X^*(\bx) > 0$ we have:
\begin{equation}\label{Formula:3}
  \epsilon = \sum_{\by} \min\BRA{W_{Y|X}(\by|\bx),\bzy^*} - e^{-R}\sum_{\by} \bzy^*
\end{equation}
and for $\bx$ such that $Q_X^*(\bx) = 0$:
\begin{equation}\label{Formula:4}
  \epsilon \leq \sum_{\by} \min\BRA{W_{Y|X}(\by|\bx),\bzy^*} - e^{-R}\sum_{\by} \bzy^*
\end{equation}
\end{theorem}

\ifFullProofs
\begin{proof}
Note that both $Q_X$ and $\bz$ range over convex compact sets and that $\gamma(Q_X,\bz)$ is a convex--concave functional (affine in $Q(\bx)$ and concave in $\bz$ by the concavity of the $\min$ function)
and $\gamma(Q_X,\bz)$ is continuous in both arguments. The existence of the saddle point and \eqref{SaddlePoint:Minimax} follow from the Fan's minimax theorem \cite{fan1953minimax}.

By the saddle point property:
  \begin{align*}
    \epsilon &= \gamma(Q_X^*,\bz^*) = \min_{Q_X} \gamma(Q_X,\bz^*) \\
  \end{align*}

Note that:
  \begin{align*}
     \gamma(Q_X,\bz) &=\sum_{\bx} Q_X(\bx)\BRA{\sum_{\by}\min\BRA{\Wyx,\bzy} - e^{-R}\sum_{\by} \bzy}
  \end{align*}
and:
  \begin{align}\label{GammaMinVal}
    \min_{Q_X} \gamma(Q_X,\bz^*) &= \min_{Q_X}\BRAs{\sum_{\bx} Q(\bx)\BRA{\sum_{\by}\min\BRA{\Wyx,\bzy^*} - e^{-R}\sum_{\by} \bzy^*} } \notag \\
    &= \min_{\bx\in\cX}\BRAs{\sum_{\by}\min\BRA{\Wyx,\bzy^*} - e^{-R}\sum_{\by} \bzy^*}
  \end{align}

hence \eqref{Formula:3} and \eqref{Formula:4} follow from the linearity of $\gamma(Q_X,\bz)$ in $Q_X$.
\end{proof}

\fi

The next theorem presents the connection between $\gamma(Q_X,\bz)$ and $\beta_{1-e^{-R}}\BRA{Q(\bx)Q(\by), Q(\bx)\Wyx}$.
\begin{theorem} \label{Theorem:Beta:Saddle}\mynewline
For any distribution $Q_X$ the following holds:
\begin{equation}\label{Formula:1}
  \max_{Q_Y} \beta_{1-e^{-R}}\BRA{Q_X\times Q_Y, Q_X W_{Y|X}} = \max_{\bz} \gamma(Q_X,\bz)
\end{equation}
  Moreover, $\bz^*$ attains the maximum in \eqref{Formula:1} if and only if for each $\by$:

\ifTwoCols
\else
\fi

\ifTwoCols
\else
  \begin{equation}\label{OptimalConditionForZ}
    Q_X\BRAs{\bx:\Wyx > \bzy^*} \leq e^{-R} \leq Q_X\BRAs{\bx:\Wyx \geq \bzy^*} 
  \end{equation}
\fi

\end{theorem}

\ifFullProofs

\begin{proof}
  \eqref{Formula:1} follows from:
  \begin{align*}
    \max_{Q_Y} \beta_{1-e^{-R}}\BRA{Q_X\times Q_Y, Q_X W_{Y|X}} &= \max_{Q_Y} \max_{\lambda}\BRA{ \sum_{\bx,\by} Q_X(\bx)\min\BRA{W_{Y|X}(\by|\bx),\lambda Q_Y(\by)} - e^{-R}\lambda } \\
                                                                       &= \max_{\bz} \BRA{ \sum_{\bx,\by} Q_X(\bx)\min\BRA{W_{Y|X}(\by|\bx),\bzy} - e^{-R}\sum_{\by} \bzy}
  \end{align*}
  where we write $\bzy = \lambda Q(\by)$ and use $\lambda = \sum_{\by} \bzy$. Note that to attain the maximum, we can restrict $\bzy \leq 1$ since $\gamma(Q_X,\bz) \leq \gamma(Q_X,\min(\bz,1))$.
To prove \eqref{OptimalConditionForZ}:
  \begin{align*}
    \gamma(Q_X,\bz) &= \sum_{\bx,\by} Q_X(\bx)\min\BRA{W_{Y|X}(\by|\bx),\bzy} - e^{-R}\sum_{\by} \bzy \\
    &= \sum_{\by} \BRA{\sum_{\bx}\min\BRA{W_{Y|X}(\by|\bx)Q_X(\bx),\bzy Q_X(\bx)} - e^{-R}\bzy} \\
    &= \sum_{\by} \BRA{\sum_{\bx}\min\BRA{W_{X|Y}(\bx|\by)Q_Y(\by),\bzy Q_X(\bx)} - e^{-R}\bzy} \\
    &= \sum_{\by} \BRA{\sum_{\bx}Q_Y(\by) \min\BRA{W_{X|Y}(\bx|\by),\frac{\bzy}{Q_Y(\by)} Q_X(\bx)} - e^{-R}\bzy} \\
    &= \sum_{\by} Q_Y(\by) \BRA{\sum_{\bx} \min\BRA{W_{X|Y}(\bx|\by),\frac{\bzy}{Q_Y(\by)} Q_X(\bx)} - e^{-R}\frac{\bzy}{Q_Y(\by)}} \\
  \end{align*}
where we assumed $Q_Y(\by) > 0$ for all $\by$ to avoid cumbersome notation.
  \begin{align*}
    \sup_{\bz} \gamma(Q_X,\bz) &= \sup_{\bz} \sum_{\by} Q_Y(\by) \BRA{\sum_{\bx} \min\BRA{W_{X|Y}(\bx|\by),\frac{\bzy}{Q_Y(\by)} Q_X(\bx)} - e^{-R}\frac{\bzy}{Q_Y(\by)}} \\
    &=  \sum_{\by} Q_Y(\by) \sup_{\bzy}  \BRA{\sum_{\bx} \min\BRA{W_{X|Y}(\bx|\by),\frac{\bzy}{Q_Y(\by)} Q_X(\bx)} - e^{-R}\frac{\bzy}{Q_Y(\by)}} \\
    &=  \sum_{\by} Q_Y(\by) \beta_{1-e^{-R}}\BRA{Q_X, W_{X|Y}} \\
  \end{align*}
  Moreover, the optimal $\bzy$ must satisfy condition \eqref{Formula:OptimalLambda}:
$$Q_X\BRAs{\bx:\frac{W(\bx|\by)}{Q(\bx)} < \frac{\bzy}{Q(\by)}} \leq 1-e^{-R} \leq Q_X\BRAs{\bx:\frac{W(\bx|\by)}{Q(\bx)} \leq \frac{\bzy}{Q(\by)}}$$
which gives \eqref{OptimalConditionForZ} after rearranging the terms.
\end{proof}
\fi

\ifFullProofs
\begin{remark}
  Combining the last theorem with \eqref{GammaMinVal} we recover the formula that appears in \cite[Proposition 14]{matthews2012linear} where it was proven by indirect arguments relying on the duality in linear programming.
\end{remark}
\fi

Theorems \ref{Theorem:Gamma:Propeties} and \ref{Theorem:Beta:Saddle} provide necessary conditions, \eqref{Formula:3},\eqref{Formula:4} and \eqref{OptimalConditionForZ} for the saddle point $Q^*_X$ and $\bz^*$. The following theorem shows that these conditions are also sufficient.

\begin{theorem}
Any distribution $Q_X^*$ and $\bz^*$ satisfy conditions \eqref{Formula:3} and \eqref{Formula:4} and \eqref{OptimalConditionForZ} is a saddle point of $\gamma(Q_X,\bz)$.
\end{theorem}

\ifFullProofs
\begin{proof}
We need to show that:  $$\gamma(Q_X^*,\bz) \leq \gamma(Q_X^*,\bz^*) \leq \gamma(Q_X,\bz^*)$$
The left hand side follows from \eqref{OptimalConditionForZ} and the right hand side from \eqref{Formula:3},\eqref{Formula:4} and the linearity in $Q_X$.
\end{proof}

\fi


\section{An Algorithm for the computation of the saddle point - High level description} \label{Sec:HighLevelDescription}

In the following sections we present our algorithm for the computation of the saddle point. We first give a high level review of the ingredients of the algorithm.

The general idea is to generate a sequence $\BRA{Q_X^{(k)},\bz^{(k)}}$ such that:
\ifTwoCols
\else
\begin{equation*}
  \gamma(Q_X^{(k)},\bz^{(k)})=\max_{\bz}\gamma(Q_X^{(k)},\bz) > \max_{\bz}\gamma(Q_X^{(k+1)},\bz)=\gamma(Q_X^{(k+1)},\bz^{(k+1)})
\end{equation*}
\fi

The initial step takes any distribution $Q_X^{(0)}$ and calculate $\bz^{(0)}$ using \eqref{OptimalConditionForZ}. Then, each iteration contains two steps as we now describe:
\subsection{Optimizing $Q_X^{(k+1)}$ for a given $\bz^{(k)}$}
Given $\bz^{(k)}$ we can find a distribution $Q_X^{(k+1)}$ that minimizes $\gamma(Q_X,\bz^{(k)})$ subject to condition \eqref{OptimalConditionForZ}. This is a linear program with $|\cX|$ variables, $2\cdot |\cY|+|\cX|$ linear inequalities, ($2\cdot |\cY|$ for \eqref{OptimalConditionForZ} and $|\cX|$ for the nonnegativity of $Q_X(\bx)$), and additional equality for $Q_X(\bx)$ to sum to 1. If:
\begin{equation*}
  \min_{Q_X}\gamma(Q_X,\bz^{(k)}) < \gamma(Q_X^{(k)},\bz^{(k)})
\end{equation*}

Then we define:
\begin{enumerate}
  \item $\bz^{(k+1)}=\bz^{(k)}$
  \item $Q_X^{(k+1)} = \argmin_{Q_X}\gamma(Q_X,\bz^{(k)})$
\end{enumerate}

We will refer to this stage as a \textbf{local linear optimization} and say that $Q_X^{(k+1)}$ is \textbf{locally optimal} given $\bz^{(k)}$.

\subsection{Improving a locally optimal solution}
When we hold a locally optimal solution $Q_X^{(k)}$, we have to change $\bz^{(k)}$ in order to improve (reduce) the current score (\ie, $\gamma(Q_X^{(k)},\bz^{(k)})$). Consider any perturbation $\mu$ on $Q_X$, \ie,
$\sum_{\bx}\mu(\bx)=0$, and let $Q_X^{\mu} = Q_X^{(k)}+\delta\mu$ where $\delta$ is small enough.\footnote{Note that when $Q_X(\bx)=0$ we must take $\mu(\bx) \geq 0$ and if $Q_X(\bx)=1$ we must take $\mu(\bx) < 0$} For $Q_X^{\mu}$, let $\bz^{\mu}$ satisfy the condition \eqref{OptimalConditionForZ} with respect to $Q_X^{\mu}$. Let:
\begin{equation}\label{OptimizingForMu}
  \eta(\mu) = \frac{  \gamma(Q_X^{\mu},\bz^{\mu})-\gamma(Q_X^{(k)},\bz^{(k)})  }{\delta}
\end{equation}
If $\min_{\mu}\eta(\mu) = 0$ then we cannot improve $Q_X^{(k)}$ and we have a \textbf{globally optimal solution}. If $\eta(\mu) < 0$ for some $\mu$, then we found an improvement of the score function and we define:
\begin{enumerate}
  \item $\bz^{(k+1)}=\bz^{\mu}$
  \item $Q_X^{(k+1)} = Q_X^{\mu}$
\end{enumerate}

In practice we will show that the problem of minimizing \eqref{OptimizingForMu} can be translated to a linear program as well (up to some regularities that we will have to handle separately), which will allow us to solve it.

\section{Improving a locally optimal solution - Details}\label{Sec:AlgorithmDetail}
In this section we describe in detail how to implement step B of the iteration, described above in high level.

Fix $Q_X$ and $\bz$ and assume the $Q_X$ is locally optimal with respect to $\bz$. Let $\mu$ be a perturbation of $Q_X$, \ie, $\mu \in \mathbb{R}^{|\cX|}$ with $\sum_{\bx}\mu(\bx)=0$.
Recall that by \eqref{OptimalConditionForZ} for each $y$ we have:
$$ Q_X\BRAs{\bx:\Wyx > \bzy} \leq e^{-R} \leq Q_X\BRAs{\bx:\Wyx \geq \bzy}$$

Assume initially that $Q_X(\bx) > 0$ for all $\bx$. We point out in the sequel where we need this assumption. When we do have zeros in the distribution $Q_X(\bx)$ we will restrict ourselves to the subset: $\BRAs{\bx \in \cX: Q_X(\bx) > 0}$. In subsection \ref{Sec:Zeros} we explain how to recover from this assumption.

\subsection{Notation}
We will make use of the following notation through this section.
\begin{enumerate}
  \item $\Ind{\Wyx \geq \bzy}$ denotes a vector, indexed by $\bx$ with $\Ind{\Wyx \geq \bzy}(\bx) = 1$ if $\Wyx \geq \bzy$ and 0 otherwise. Define $\Ind{\Wyx > \bzy}$ likewise.
  \item $\mu^T\cdot L$ is the scalar product between the vectors $\mu$ and $L$, \ie: $\mu^T\cdot L = \sum_{\bx}\mu(\bx)L(\bx)$.
\end{enumerate}
\subsection{Phase I: Changing $\bz$ to achieve strict inequality on the left hand side of \eqref{OptimalConditionForZ}}
Throughout, we assume that:
\begin{align*}
Q_X\BRAs{\bx:\Wyx > \bzy} &< e^{-R} \\
&\leq Q_X\BRAs{\bx:\Wyx \geq \bzy}
\end{align*}
\ie, we have strict inequality on the left hand side of \eqref{OptimalConditionForZ}. If this is not the case, we can change $\bzy$ until this is valid for all $y$.
\ifFullProofs

If $Q_X\BRAs{\bx:\Wyx > \bzy} = e^{-R}$, Let:
$$ \bx_y = \argmin_{\bx}\BRAs{\Wyx: \Wyx > \by_z, Q_X(\bx)>0}$$
Then:
\begin{itemize}
  \item $Q_X\BRAs{\bx:\Wyx > \bzy} = Q_X\BRAs{\bx:\Wyx \geq W_{Y|X}(\by|\bx_y)}$
  \item $Q_X(\bx_y) > 0$
  \item $Q_X\BRAs{\bx:\Wyx > W_{Y|X}(\by|\bx_y)} < e^{-R}$ since $Q_X(\bx_y) > 0$.
\end{itemize}

Replacing $\bzy$ with $W_{Y|X}(\by|\bx_y)$ we have strict inequality on the left hand side in \eqref{OptimalConditionForZ} and we haven't changed the local optimality since the optimality condition \eqref{OptimalConditionForZ} still holds by construction.
\fi

\subsection{Phase II: Compute Alternative $\bz$ with strict inequality on the right hand side of \eqref{OptimalConditionForZ}}
Following the same reasoning, we can find $\bz^l_y \leq \bzy$ that also satisfy \eqref{OptimalConditionForZ} with the following additional properties:

\ifTwoCols
\else
\begin{itemize}
  \item If $Q_X\BRAs{\bx:\Wyx > \bzy} < e^{-R} < Q_X\BRAs{\bx:\Wyx \geq \bzy}$ then $\bzy^l=\bzy$.
  \item $Q_X\BRAs{\bx:\Wyx > \bzy^l} \leq e^{-R} < Q_X\BRAs{\bx:\Wyx \geq \bzy^l}$
  \item If $Q_X\BRAs{\bx:\Wyx > \bzy^l} = e^{-R} = Q_X\BRAs{\bx:\Wyx \geq \bzy}$ then: $\Ind{\Wyx > \bzy^l}=\Ind{\Wyx \geq \bzy}$.
\end{itemize}

In order for the last equality to hold we must assume that: $Q_X(\bx) > 0$ for all $\bx$.
\fi

\subsection{Phase III: Compute $\bz^{\mu}$}
Let $Q_X^{\mu}=Q_X+\delta\cdot\mu$ where $\delta$ is sufficiently small. Recall that we must find $\bz^{\mu}$ that satisfies the condition \eqref{OptimalConditionForZ} with respect to $Q_X^{\mu}$.
From:
\ifTwoCols
\else
$$ Q_X^{\mu}\BRAs{\bx:\Wyx > \bzy} = Q_X\BRAs{\bx:\Wyx > \bzy}+\delta\mu^T\cdot \Ind{\bx:\Wyx > \bzy}$$
\fi
we always have $$Q_X^{\mu}\BRAs{\bx:\Wyx > \bzy} < e^{-R}$$
for sufficiently small $\delta$ and:

\ifTwoCols
\else
$$ Q_X^{\mu}\BRAs{\bx:\Wyx \geq \bzy} \geq e^{-R} \Leftrightarrow \mu^T\cdot \Ind{\bx:\Wyx \geq \bzy} \geq 0$$
\fi

Hence when $\mu^T\cdot \Ind{\bx:\Wyx \geq \bzy} < 0$ we must change $\bzy$ since it does not satisfy condition \eqref{OptimalConditionForZ} anymore.
Since:
$$Q_X^{\mu}\BRAs{\bx:\Wyx \geq \bzy^l} > e^{-R}$$
for sufficiently small $\delta$ and:
\ifTwoCols
\else
$$ Q_X^{\mu}\BRAs{\bx:\Wyx > \bzy^l} \leq e^{-R} \Leftrightarrow \mu^T\cdot \Ind{\bx:\Wyx > \bzy^l} \leq 0$$
\fi

Now, from $ \Ind{\bx:\Wyx > \bzy^l} = \Ind{\bx:\Wyx \geq \bzy}$ we have:
$$ \mu^T\cdot \Ind{\bx:\Wyx > \bzy^l} = \mu^T\cdot \Ind{\bx:\Wyx \geq \bzy}$$
and when $\mu^T\cdot \Ind{\bx:\Wyx \geq \bzy} < 0$ we can take $\bzy^l$.

To summarize, let:
\begin{equation}\label{Def:Z_mu}
\bz^{\mu}_y =
\left\{
	\begin{array}{ll}
		\bzy    & \mbox{if } \mu^T \cdot \Ind{\bx: \Wyx \geq \bzy} \geq 0 \\
		\bz^l_y    & \mbox{if } \mu^T \cdot \Ind{\bx: \Wyx \geq \bzy}    < 0
	\end{array}
\right.
\end{equation}
Then $\bz^{\mu}$ satisfies \eqref{OptimalConditionForZ} with respect to $Q_X^{\mu}$ for $\delta$ sufficiently small.

\subsection{Computation of $\gamma(Q_X^{\mu},\bz^{\mu})$}

Let:
\begin{equation}\label{Def:FrLinearPart}
    \eta(\mu,\bz) \triangleq \sum_{\bx,\by}\mu(\bx)\min\BRA{\Wyx, \bzy}
\end{equation}
We have:
\begin{align*}
  \gamma(Q_X^{\mu},\bz^{\mu}) &= \sum_{\bx,\by}\BRA{Q_X(\bx)+\delta\mu(\bx)}\min\BRA{\Wyx, \bz^{\mu}_y}\\
                              &-e^{-R}\sum_{\by}\bz^{\mu}_y \\
                              &= \gamma(Q_X,\bz^{\mu})+\delta \sum_{\bx,\by}\mu(\bx)\min\BRA{\Wyx, \bz^{\mu}_y} \\
                              &= \gamma(Q_X,\bz^{\mu})+\delta \eta(\mu,\bz^{\mu})
\end{align*}
Since $\bz^{\mu}$ also satisfies \eqref{OptimalConditionForZ} with respect to $Q_X$, $\gamma(Q_X,\bz^{\mu}) = \gamma(Q_X,\bz)$ and:
\begin{align*}
  \frac{\gamma(Q_X^{\mu},\bz^{\mu})-\gamma(Q_X,\bz)}{\delta} &= \frac{\gamma(Q_X^{\mu},\bz^{\mu})-\gamma(Q_X,\bz^{\mu})}{\delta} \\
  &= \eta(\mu,\bz^{\mu})
\end{align*}

\ifFullProofs

and:
\begin{align*}
  \eta(\mu,\bz^{\mu})-\eta(\mu,\bz) &= \sum_{\bx,\by} \mu(\bx)\BRA{\min\BRA{\Wyx, \bzy^{\mu}}-\min\BRA{\Wyx, \bzy}} \\
  &=\sum_{\by:\mu^T \cdot \Ind{\Wyx \geq \bzy} < 0} \sum_{\bx} \BRA{\min\BRA{\Wyx, \bzy^l}-\min\BRA{\Wyx, \bzy}} \\
  &\overset{(a)}{=} \sum_{\by:\mu^T \cdot \Ind{\Wyx \geq \bzy} < 0} \BRA{\bzy^l-\bzy}\mu^T\cdot\Ind{\bx: \Wyx \geq \bzy} \\
  &= \sum_{\by} \BRA{\bzy^l-\bzy}\mu^T\cdot\Ind{\Wyx \geq \bzy} \Ind{\mu^T \cdot \Ind{\Wyx \geq \bzy} < 0}\\
\end{align*}
where (a) follows from:
\begin{align}\label{MinEquivalent}
  &\sum_{\bx}\BRA{\min\BRA{\Wyx, \bzy^l}-\min\BRA{\Wyx, \bzy}} \\
  &= \sum_{\bx:\Wyx > \bzy^l}\bzy^l+\sum_{\bx:\Wyx \leq \bzy^l}\Wyx-\sum_{\bx:\Wyx \geq \bzy}\Wyx-\sum_{\bx:\Wyx < \bzy}\bzy \notag\\
  &=\BRA{\bzy^l-\bzy}\mu^T\cdot\Ind{\bx: \Wyx \geq \bzy} \notag
\end{align}
 since $\Ind{\bx:\Wyx > \bzy^l} = \Ind{\bx:\Wyx \geq \bzy}$ and also $\Ind{\bx:\Wyx \leq \bzy^l} = \Ind{\bx:\Wyx < \bzy}$.
To sum until here:
\begin{equation}\label{OptForMu}
  \eta(\mu,\bz^{\mu}) = \eta(\mu,\bz) - \sum_{\by} \BRA{\bzy-\bzy^l}\mu^T\cdot\Ind{\Wyx \geq \bzy} \Ind{\mu^T \cdot \Ind{\Wyx \geq \bzy} < 0}
\end{equation}

And we want to optimize $\eta(\mu,\bz^{\mu})$ with respect to $\mu$.
\fi

\subsection{Optimize for $\mu$}

Let define:
\begin{itemize}
  \item $b(\bx) = \sum_{\by}\min\BRA{W_{Y|X}(\by|\bx), \bzy} $ so that: $\eta(\mu,\bz) = \mu^T\cdot b$
  \item $a_y = \Ind{\bx: \Wyx \geq \bzy}$
  \item $\alpha_y = \bzy-\bzy^l \geq 0$
\end{itemize}
Then:
\begin{equation}\label{OptForMuAbstractForm}
  \eta(\mu,\bz^{\mu}) = \eta(\mu)= \mu^T\cdot\BRA{b-\sum_{\by} \alpha_y a_y \Ind{\mu^T \cdot a_y < 0}}
\end{equation}

\ifFullProofs
In appendix \ref{App:FarkasLemma}
\fi
we prove the following two lemmas. The first shows how to translate the problem of minimizing $\eta(\mu)$ into a linear program. We provide these lemmas here using the notation used in this section. (\ie, index the vectors with $y$)
\begin{lemma}\label{Lemma:EquivalentLP:Y}
Let:
\begin{equation}\label{Def:Eta}
  \eta(\mu) = \mu^T \cdot \BRA{b - \sum_y \alpha_y a_y \Ind{\mu^T \cdot a_y < 0} }
\end{equation}
Then minimization of $\eta(\mu)$ subject to $\mu^T\cdot\ones=0$ is equivalent to the following linear program:
\begin{equation}\label{EquivalentLP}
  \mbox{min}\ColVec{\mu}{z}^T \cdot \ColVec{b}{\alpha} \mbox{s.t.} \ColVec{\mu}{z}^T \cdot \BRA{ \begin{array}{cc} A & 0 \\ I & I \\ \end{array}} \geq 0, \mu^T\cdot\ones=0
\end{equation}
where $A$ is the matrix with columns $a_y$, $\alpha$ is a vector with entries $\alpha_y$, and $\ones$ is the all-one vector.
\end{lemma}
The next lemma provides necessary and sufficient conditions for $\mu=0$ to be the optimal minimizer of $\eta(\mu)$.
\begin{lemma}[Generalized Farkas]\label{Lemma:GeneralFarkas:Y}
    Let $a_y \in \mathbb{R}^n, y\in\cY$, $b \in \mathbb{R}^{|\cX|}$ and $\alpha_y \geq 0$.
    Then
\begin{equation}\label{Cond1}
  \mu^T \cdot \BRA{b - \sum_y \alpha_y a_y \Ind{\mu^T \cdot a_y < 0} } \geq 0
\end{equation}
   for all $\mu \in \mathbb{R}^{|\cX|}$ such that $\mu^T\cdot\ones=0$ if and only if:
\begin{equation}\label{Lemma:OptimalCondGeneralFarkasWithYs}
  b = \sum_{j} \lambda_y a_y+\tau\ones, 0 \leq \lambda_y \leq \alpha_y, \tau\in\mathbb{R}
\end{equation}
\end{lemma}

If $\eta(\mu) < 0$ then we have found an improvement of the score and we can keep on going to find a new locally optimal solution.

\subsection{The case where $\min_{\mu}\eta(\mu) = 0$}\label{SubSec:LocalImprovmentOptimal}
If $\eta(\mu) = 0$ is the minimal value, then we cannot improve on the current solution using perturbation that consider non-zeros elements of $Q_X(\bx)$. (The case where there are zeros in $Q_X(\bx)$ is discussed in subsection \ref{Sec:Zeros}).

Let us show that indeed in this case we reached the optimal solution, i.e., we can recover the conditions \eqref{Formula:3} and \eqref{Formula:4}.

Define $\bz^{o}$ by: $\bz^{o}_y = \bzy-\lambda_y$. Then:
\begin{enumerate}
  \item $\bzy^l \leq \bz^{o}_y \leq \bzy$
  \item $b^o(\bx) = \sum_{\by}\min\BRA{W_{Y|X}(\by|\bx), \bzy^o} = \tau$, \ie\  $b^o = \tau\ones$
\end{enumerate}

\ifFullProofs
The last equality follows from:
\begin{align*}
  \mu^T\cdot\BRA{b^o-b} &= \sum_{\by} \sum_{\bx} \mu(\bx)\BRA{\min\BRA{W_{Y|X}(\by|\bx), \bzy^o}-\min\BRA{W_{Y|X}(\by|\bx), \bzy}} \\
  &\overset{(a)}{=} \sum_{\by} \BRA{\bzy^o-\bzy}\mu^T\cdot\Ind{\bx: \Wyx \geq \bzy} \\
  &= -\sum_{\by}\lambda_y \mu^T\cdot a_y\\
\end{align*}
where (a) follows from the same reasoning as \eqref{MinEquivalent}.
Hence:
$$ b^o=b-\sum_{\by} \lambda_y a_y =\tau e$$
\fi


\subsection{Zeros in $Q_X(\bx)$}\label{Sec:Zeros}
\ifFullProofs

Let $Q_X,\bz$ be such that:
$$ Q_X\BRAs{\bx: \Wyx > \bzy} \leq e^{-R} \leq Q_X\BRAs{\bx: \Wyx \geq \bzy}$$
$$ \epsilon = \sum_{\by} \min\BRA{W_{Y|X}(\by|\bx),\bzy} - e^{-R}\sum_{\by} \bzy $$ for all $\bx$ with $Q_X(\bx) > 0$, and:
$$ \epsilon >    \sum_{\by} \min\BRA{W_{Y|X}(\by|\bx_1),\bzy} - e^{-R}\sum_{\by} \bzy $$ for some $\bx_1$ with $Q_X(\bx_1) = 0$.
We also assume that $Q_X$ is locally optimal, which means that we cannot improve the score by running a local linear program. Obviously, we cannot argue that the optimality condition \eqref{Formula:4} holds.

For any perturbation with $\mu(\bx_1) > 0$, we must have that at least one of the linear inequality constraints is violated. Equivalently, we can say:
For any perturbation that does not violate the linear inequality constraint, we must have $\mu(\bx_1) \leq 0$.

\begin{itemize}
  \item From $Q_X^{\mu}\BRAs{\bx: \Wyx > \bzy} = Q_X\BRAs{\bx: \Wyx > \bzy} + \mu^T\cdot \Ind{\Wyx > \bzy} $, If $Q_X\BRAs{\bx: \Wyx > \bzy}=e^{-R}$ then in order not to violate the linear inequality we must have: $\mu^T\cdot \Ind{\Wyx > \bzy} \leq 0$
  \item From $Q_X^{\mu}\BRAs{\bx: \Wyx \geq \bzy} = Q_X\BRAs{\bx: \Wyx \geq \bzy} + \mu^T\cdot \Ind{\Wyx \geq \bzy} $, If $Q_X\BRAs{\bx: \Wyx \geq \bzy}=e^{-R}$ then in order not to violate the linear inequality we must have: $\mu^T\cdot \Ind{\Wyx \geq \bzy} \geq 0$
  \item $\mu$ must satisfy: $\mu^T\cdot \ones = 0$.
\end{itemize}
By Farkas lemma \eqref{Lemma:Farkas} we must have:
$$ \delta_{\bx_1} = \sum_{\by:Q_X\BRAs{\bx: \Wyx > \bzy}=e^{-R}} \lambda^l_y \Ind{\Wyx > \bzy} + \sum_{\by:\by:Q_X\BRAs{\bx: \Wyx \geq \bzy}=e^{-R}} \lambda^h_y \Ind{\Wyx \geq \bzy} + \alpha\ones $$
with $\lambda^l_y \geq 0$ and $\lambda^h_y \leq 0$ and $\delta_{\bx_1}$ is the vector with 1 at $\bx_1$ and 0 otherwise.

At this point we can use these $\lambda$s by adding them to $\bzy$ in order to increase score at $\bx_1$ up to the other scores and meet the conditions \eqref{Formula:4} along the same lines as \ref{SubSec:LocalImprovmentOptimal}. Note that we might not be able to do this in a single step. Moreover, we have to do this process for each variable with $Q_X(\bx) = 0$ and lower score than the global score we have.

\fi

\section{Summary}
In this paper we have studied the functional properties of the minimax-converse for a fixed rate. The existence of a saddle point was proved, necessary and sufficient conditions were derived and an algorithm for the computation of the saddle point was presented. For the DMC case, the algorithm can be modified to incorporate additional linear constraints (\ie, input and output distribution that are uniform on types) and this results in a polynomial time algorithm for the computation of the saddle point. The saddle point distribution can be used to optimize the random coding argument (\eg, \cite{ElkayamITW2015}).

\appendices


\ifFullProofs

\section{Proof of lemma \ref{Lemma:BinaryHyp}}\label{App:BinaryHypLemma}

\subsection{Proof of \eqref{Formula:Beta}}
  Let $\lambda, \delta$ be the thresholds for the optimal test, and let:
  $$A = \BRAs{w:\frac{Q(w)}{P(w)} < \lambda} $$
  $$B = \BRAs{w:\frac{Q(w)}{P(w)} = \lambda} $$
  Then:
  \begin{equation}\label{Def:Alpha}
    \alpha = P(A)+\delta P(B)
  \end{equation}
  And:
  \begin{equation}\label{Def:Beta}
    \beta = Q(A)+\delta Q(B)
  \end{equation}
  Multiply \eqref{Def:Alpha} by $\lambda$, subtract \eqref{Def:Beta} and use $Q(B) = \lambda P(B)$:
  $$ \beta-\lambda\alpha = Q(A)-\lambda P(A)$$

  On the other hand:
  \begin{align*}
    \sum_{w\in W} \min\BRA{Q(w),\lambda P(w)} &= \sum_{w \in A}Q(w) + \sum_{w \in A^c}\lambda P(w)\\
    &= Q(A) + \lambda (1-P(A)) \\
    &= Q(A) - \lambda P(A)+\lambda \\
    &=\beta-\lambda\alpha+\lambda\\
  \end{align*}
  Thus:
  $$ \beta = \sum_{w\in W} \min\BRA{Q(w),\lambda P(w)} -\lambda(1-\alpha)$$

\subsection{proof of the sup formula (smaller $\lambda$)}

Note that the optimal $\lambda$ satisfies the following:
\begin{equation}\label{OptimalLambda}
  P\BRAs{w:\frac{Q(w)}{P(w)} \geq \lambda} \geq 1-\alpha \geq P\BRAs{w:\frac{Q(w)}{P(w)} > \lambda}
\end{equation}
  Let $\lambda_1 < \lambda$:
\begin{align*}
  & \sum_{w\in W} \min\BRA{Q(w),\lambda_1 P(w)} - \sum_{w\in W} \min\BRA{Q(w),\lambda P(w)} \\
         &= \sum_{w\in W:\lambda_1 P(w) < Q(w) < \lambda P(w)} \BRA{\lambda_1 P(w) - Q(w)}+\BRA{\lambda_1  - \lambda }\sum_{w\in W:\lambda P(w) \leq Q(w) } P(w) \\
         &\overset{(a)}{\leq} \BRA{\lambda_1  - \lambda }\sum_{w\in W:\lambda P(w) \leq Q(w) } P(w) \\
         &= \BRA{\lambda_1  - \lambda }P\BRAs{w: \frac{Q(w)}{P(w)} \geq \lambda } \\
         &\overset{(b)}{\leq} \BRA{\lambda_1  - \lambda }(1-\alpha)
\end{align*}
where (a) follow from: $\lambda_1 P(w) - Q(w)< 0$, (b) follow from $\lambda_1  - \lambda < 0$ and $P\BRAs{w: \frac{Q(w)}{P(w)} \geq \lambda } \geq 1-\alpha$.
Rearranging the terms:
\begin{equation*}
  \sum_{w\in W} \min\BRA{Q(w),\lambda_1 P(w)} -\lambda_1(1-\alpha)\leq \sum_{w\in W} \min\BRA{Q(w),\lambda P(w)}-\lambda(1-\alpha)
\end{equation*}
If $\lambda_1$ does not satisfy the condition \eqref{Formula:OptimalLambda}, then:
\begin{itemize}
  \item If $P\BRAs{w:\frac{Q(w)}{P(w)} \leq \lambda_1} < P\BRAs{w:\frac{Q(w)}{P(w)} < \lambda}$, then we are finished because there exist $w_0$ with $P(w_0) > 0$, $\frac{Q(w_0)}{P(w_0)} < \lambda$, and $\frac{Q(w_0)}{P(w_0)} > \lambda_1$, which gives strict inequality in (a) above.
  \item If $P\BRAs{w:\frac{Q(w)}{P(w)} \leq \lambda_1} = P\BRAs{w:\frac{Q(w)}{P(w)} < \lambda}$ then $P\BRAs{w:\frac{Q(w)}{P(w)} \leq \lambda_1} < \alpha$ and we have strict inequality $P\BRAs{w:\frac{Q(w)}{P(w)} < \lambda} < \alpha$, which leads to a strict inequality in (b) above.
\end{itemize}

\subsection{Proof of the sup formula (greater $\lambda$)}
For $\lambda_1 > \lambda$ we have:

\begin{align*}
  \sum_{w\in W} \min\BRA{Q(w),\lambda_1 P(w)} &=    Q\BRAs{w:Q(w) < \lambda P(w)}+Q\BRAs{w:\lambda P(w) \leq Q(w) \leq \lambda_1 P(w)}+\lambda_1 P\BRAs{w:Q(w) > \lambda_1 P(w)} \\
                                              &\overset{(a)}{\leq} Q\BRAs{w:Q(w) < \lambda P(w)}+\lambda_1 P\BRAs{w:\lambda P(w) \leq Q(w) \leq \lambda_1 P(w)}+\lambda_1 P\BRAs{w:Q(w) > \lambda_1 P(w)} \\
                                              &=    Q\BRAs{w:Q(w) < \lambda P(w)}+\lambda_1 P\BRAs{w:Q(w) \geq \lambda P(w)}
\end{align*}
where (a) follow upper bounding $Q(w)$ with $\lambda_1 P(w)$.
\begin{align*}
  & \sum_{w\in W} \min\BRA{Q(w),\lambda_1 P(w)} - \sum_{w\in W} \min\BRA{Q(w),\lambda P(w)} \\
                              &\leq Q\BRAs{w:Q(w) < \lambda P(w)}+\lambda_1 P\BRAs{w:Q(w) \geq \lambda P(w)}-Q\BRAs{w:Q(w) < \lambda P(w)}-\lambda P\BRAs{w:Q(w) \geq \lambda P(w)} \\
                              &= (\lambda_1-\lambda) P\BRAs{w:Q(w) \geq \lambda P(w)} \\
                              &\leq \BRA{\lambda_1-\lambda} (1-\alpha)
\end{align*}
Since $\lambda_1-\lambda > 0$ and $P\BRAs{w:Q(w) \geq \lambda P(w)} \leq 1-\alpha$, we have:
\begin{equation*}
  \sum_{w\in W} \min\BRA{Q(w),\lambda_1 P(w)} -\lambda_1(1-\alpha) \leq \sum_{w\in W} \min\BRA{Q(w),\lambda P(w)}-\lambda(1-\alpha)
\end{equation*}

If $\lambda_1$ does not satisfy the condition \eqref{Formula:OptimalLambda}, then $P\BRAs{w:\frac{Q(w)}{P(w)} < \lambda} < P\BRAs{w:\frac{Q(w)}{P(w)} < \lambda_1}$ and we are finished because there exist $w_0$ with $P(w_0) > 0$, $\frac{Q(w_0)}{P(w_0)} \geq \lambda$, and $\frac{Q(w_0)}{P(w_0)} < \lambda_1$, which gives strict inequality in (a) above.

\else
\fi


\ifFullProofs

\section{Generalized Farkas Lemma } \label{App:FarkasLemma}
\begin{lemma}[Farkas]\label{Lemma:Farkas}
  Let $a_i \in \mathbb{R}^n, i=1,...,m$ and  $b \in \mathbb{R}^n$. If for all $\mu \in \mathbb{R}^n$ such that $\mu^T \cdot a_i \geq 0$ implies $\mu^T \cdot b \geq 0$, then $b = \sum_{j} \lambda_j a_j$ with $\lambda_j \geq 0$.
\end{lemma}

We need to prove the following generalization of this result:
\begin{lemma}[Generalized Farkas]\label{Lemma:GeneralFarkas}
    Let $a_i \in \mathbb{R}^n, i=1,...,m$, $b \in \mathbb{R}^n$ and $\alpha_j \geq 0$.
    Assume that for all $\mu \in \mathbb{R}^n$ such that $\mu^T\cdot C \geq 0$,
\begin{equation}\label{Cond1}
  \mu^T \cdot \BRA{b - \sum_j \alpha_j a_j \Ind{\mu^T \cdot a_j < 0} } \geq 0
\end{equation}
   Then:  $b = \sum_{j} \lambda_j a_j + C\cdot \tau$ with $0 \leq \lambda_j \leq \alpha_j$.
\end{lemma}
We defer the proof of the lemma after proving the following:
\begin{lemma}\label{Lemma:EquivalentLP}
Let:
\begin{equation}\label{Def:Eta}
  \eta(\mu) = \mu^T \cdot \BRA{b - \sum_j \alpha_j a_j \Ind{\mu^T \cdot a_j < 0} }
\end{equation}
Then minimization of $\eta(\mu)$ such that $\mu^T\cdot C \geq 0$, is equivalent to the following linear program:
\begin{equation}\label{EquivalentLP}
  \mbox{min}\ColVec{\mu}{z}^T \cdot \ColVec{b}{\alpha} \mbox{s.t.} \ColVec{\mu}{z}^T \cdot \LPMat \geq 0
\end{equation}
where $A$ is the matrix with columns $a_i$, $\lambda$ is a vector with entries $\lambda_i$, and $\alpha$ is a vector with entries $\alpha_i$.
\end{lemma}

\begin{remark}
Obviously, 0 is an admissible solution. when there exists $\mu$ with $\eta(\mu) < 0$ then $\eta$ is not bounded from below since we can multiply the solution by any positive factor. When the solution is bounded from below then it must be 0.
\end{remark}
\begin{proof}
For each $\mu$ such that $\mu^T\cdot C \geq 0$, let $z_{opt}(\mu) = \max\BRA{0, -\mu^T\cdot A}$. Then:
$$ \ColVec{\mu}{z_{opt}(\mu)}^T \cdot \LPMat \geq 0 $$
and:
\begin{align*}
  \ColVec{\mu}{z_{opt}}^T \cdot \ColVec{b}{\alpha} &= \mu^T \cdot b -\sum_{j:\mu^T\cdot a_j < 0} \alpha_j \BRA{\mu^T\cdot a_j}\\
  &= \mu^T \cdot \BRA{b - \sum_j \alpha_j a_j \Ind{\mu^T \cdot a_j < 0} }
\end{align*}
On the other hand, if:
$$\ColVec{\mu}{z}^T \cdot \LPMat \geq 0$$
then: $$ z \geq z_{opt}(\mu)$$
and:
$$ \ColVec{\mu}{z}^T \cdot \ColVec{b}{\alpha} \geq \ColVec{\mu}{z_{opt}(\mu)}^T \cdot \ColVec{b}{\alpha} = \mu^T \cdot \BRA{b - \sum_j \alpha_j a_j \Ind{\mu^T \cdot a_j < 0} }$$
since $\alpha \geq 0$.
\end{proof}

\begin{proof}[Proof of Lemma \ref{Lemma:GeneralFarkas}]
By Lemma \ref{Lemma:EquivalentLP}, the problem is equivalent to the linear program \eqref{EquivalentLP}. If 0 is the minimal solution, then this is equivalent to:
\begin{equation}\label{Cond2}
  \ColVec{\mu}{z}^T \cdot \LPMat \geq 0 \Rightarrow  \ColVec{\mu}{z}^T \cdot \ColVec{b}{\alpha} \geq 0
\end{equation}
Now, the standard Farkas lemma \eqref{Lemma:Farkas}, this implies that there exists $\lambda = \TripleColVec{\lambda_1}{\lambda_2}{\tau} \geq 0$ such that:
$\LPMat \cdot \lambda = \ColVec{b}{\alpha}$
which is equivalent to: $A \cdot \lambda_1 + C\cdot \tau = b$ and $\lambda_1+\lambda_2 = \alpha$, which together give $0 \leq \lambda_1 \leq \alpha$ as needed.
\end{proof}

\begin{remark}
Note that we can add equality constraints on $\mu$ by adding two inequality constraints. For an equality constraint this results in an additional vector added to $b$ without any restriction on their coefficient. Specifically, in our case we have the additional constraint that $\sum_{\bx}\mu(\bx)=0$, which is equivalent to $\mu^T\cdot\ones = 0$, where $\ones$ is the vector of all ones. In Lemma \ref{Lemma:GeneralFarkas:Y} we obtain:
$$  b = \sum_{y} \lambda_y a_y + \tau\ones, \lambda_j \geq 0 $$
where we don't have restrictions on $\tau$ (using the notation there).
\end{remark}

\else
\fi

\ifFullProofs

\section{Modification for DMC } \label{App:DMC}
In this section we assume the reader is familiar with the \emph{method of types} \cite{csiszar1998method,csiszár2011information}.  We use standard type notation, \eg \cite{csiszár2011information}. Specifically, for a fixed $n$:
\begin{itemize}
  \item $\bx \in \cX^n$ and $\by \in \cY^n$
  \item $P_{\bx}$ denotes the empirical distribution of the sequence $\bx \in \cX^n$. $P_{\bx,\by}$ denotes the empirical distribution of the sequence $\BRA{\bx,\by} \in \BRA{\cX\times\cY}^n$
  \item $\Tx$ denotes the type class of the sequence $\bx$, \ie:
  $$T_{\bx} = \BRAs{\bx'\in \cX^n: P_{\bx'}=P_{\bx}}$$
  \item $T_{\bx|\by}$ is the conditional type class of $\bx$ given $\by$, \ie:
  $$T_{\bx|\by} = \BRAs{\bx'\in \cX^n: P_{\bx',\by}=P_{\bx,\by} }$$
  \item $|\cdot|$ denote the size of a set, \eg, $|\Tx|$,$|T_{\bx|\by}|$
\end{itemize}

For DMC, we know from \cite[Theorem 20]{polyanskiy2013saddle} that we can restrict both the input and output distributions, as $Q_X(\bx)$ and $Q_Y(\by)$ to be uniform on types. Using the same argument for $\gamma(Q_X,\bz)$, \ie, the convexity and concavity with respect to $Q_X$ and $\bz$ shows that we can also prove that $Q_X(\bx)$ and $\bz$ are uniform over type. In this appendix we provide the necessary modification for the algorithm needed. Specifically, for each input type class $\Tx$ let $\lambda_{\Tx}=Q_X\BRA{\Tx}$, \ie\ $\lambda_{\Tx}$ is the weight of the type class $\Tx$. We have:
$$ \sum_{\Tx} \lambda_{\Tx} = 1$$
and:
\begin{equation}\label{Def:QofLambda}
  Q_X(\bx) = \frac{\lambda_{\Tx}}{|\Tx|}
\end{equation}
where \eqref{Def:QofLambda} is by the uniform type assumption. We also assume that $\bzy$ is fixed for each $\by'\in T_{\by}$, \ie\ $\bzy = \bz_{\Ty}$.
The algorithm is modified to calculate the score using $\lambda_{\Tx}$ and $\bzy$ instead of $Q_X(\bx)$ and $\bzy$. The linear inequality and the score function has to be modified to incorporate $\lambda_{\Tx}$ and $\bz_{\Ty}$ instead of $Q_X(\bx)$ and $\bzy$.
For the linear inequality:
\begin{align*}
  Q_X\BRAs{\Wyx > \bzy} &= \sum_{\bx:\Wyx > \bzy} Q_X(\bx) \\
  &= \sum_{\bx:\Wyx > \bzy} \frac{\lambda_{\Tx}}{|T_x|} \\
  &\overset{(a)}{=} \sum_{T_{\bx|\by}:\Wyx > \bzy} \lambda_{\Tx}\frac{|T_{\bx|\by}|}{|T_x|} \\
\end{align*}
where in (a) we sum over the conditional type of $\bx$ given $\by$, which satisfies the condition. The condition with $\geq$ instead of $>$ is similar. The score function:
\begin{align*}
\gamma(Q_X,\bz) &= \sum_{\bx,\by} Q_X(\bx)\min\BRA{W_{Y|X}(\by|\bx),\bzy} - e^{-R}\sum_{\by} \bzy \\
&\overset{(b)}{=} \sum_{T_{\bx,\by}} |T_{\bx,\by}|Q_X(\bx)\min\BRA{W_{Y|X}(\by|\bx),\bzy} - e^{-R}\sum_{\Ty}|\Ty| \bzy \\
&= \sum_{T_{\bx,\by}} |T_{\bx,\by}|\frac{\lambda_{\Tx}}{|\Tx|}\min\BRA{W_{Y|X}(\by|\bx),\bzy} - e^{-R}\sum_{\Ty}|\Ty| \bzy \\
&\overset{(c)}{=} \sum_{T_{\bx,\by} } |T_{\by|\bx}|\lambda_{\Tx}\min\BRA{W_{Y|X}(\by|\bx),\bzy} - e^{-R}\sum_{\Ty}|\Ty| \bzy \\
\end{align*}
where (b) follows by summing over all $(\bx,\by)$ in the type class $T_{\bx,\by}$, since $\sum_{\bx} Q_X(\bx)\min\BRA{W_{Y|X}(\by|\bx),\bzy}$ is constant over the type class, and the same argument for the second sum (c) follows since $\frac{|T_{\bx,\by}|}{|\Tx|}=|T_{\by|\bx}|$.

A few comments are in order:
\begin{remark}
\mynewline
  \begin{itemize}
    \item Since we expect $e^{-R}$ to be small, this suggests that calculations should be done in the log domain. This is left for further research.
    \item It is well known that the linear programs are harder when degeneracy occurs. This follows in our case too; had we assumed that no degeneracy occurs, some simplifications are possible. However, since we are interested in small examples, simulation results show that degeneracy does occur and we have to handle these cases as well.
    \item Incremental algorithm starting from large $R$ for which the uniform distribution is optimal and reducing $R$ while keeping optimality of the distribution $Q_X$ through small correction to the distribution.
  \end{itemize}
\end{remark}

\else
\fi

\bibliographystyle{IEEEtran}
\bibliography{bib}

\begin{thebibliography}{1}
\providecommand{\url}[1]{#1}
\csname url@samestyle\endcsname
\providecommand{\newblock}{\relax}
\providecommand{\bibinfo}[2]{#2}
\providecommand{\BIBentrySTDinterwordspacing}{\spaceskip=0pt\relax}
\providecommand{\BIBentryALTinterwordstretchfactor}{4}
\providecommand{\BIBentryALTinterwordspacing}{\spaceskip=\fontdimen2\font plus
\BIBentryALTinterwordstretchfactor\fontdimen3\font minus
  \fontdimen4\font\relax}
\providecommand{\BIBforeignlanguage}[2]{{%
\expandafter\ifx\csname l@#1\endcsname\relax
\typeout{** WARNING: IEEEtran.bst: No hyphenation pattern has been}%
\typeout{** loaded for the language `#1'. Using the pattern for}%
\typeout{** the default language instead.}%
\else
\language=\csname l@#1\endcsname
\fi
#2}}
\providecommand{\BIBdecl}{\relax}
\BIBdecl

\bibitem{polyanskiy2010channel}
Y.~Polyanskiy, H.~V. Poor, and S.~Verd{\'u}, ``Channel coding rate in the
  finite blocklength regime,'' \emph{Information Theory, IEEE Transactions on},
  vol.~56, no.~5, pp. 2307--2359, 2010.

\bibitem{polyanskiy2013saddle}
Y.~Polyanskiy, ``Saddle point in the minimax converse for channel coding,''
  \emph{Information Theory, IEEE Transactions on}, vol.~59, no.~5, pp.
  2576--2595, 2013.

\bibitem{ElkayamITW2015}
\BIBentryALTinterwordspacing
N.~Elkayam and M.~Feder, ``Achievable and converse bounds over a general
  channel and general decoding metric,'' \emph{arXiv preprint arXiv:1411.0319},
  2014. [Online]. Available:
  \url{http://www.eng.tau.ac.il/~elkayam/FiniteBlockLen.pdf}
\BIBentrySTDinterwordspacing

\bibitem{ElkayamBinaryTest}
\BIBentryALTinterwordspacing
------. (2016) Variational formulas for the power of the binary hypothesis
  testing problem with applications. [Online]. Available:
  \url{http://www.eng.tau.ac.il/~elkayam/Binary_ISIT.pdf}
\BIBentrySTDinterwordspacing

\bibitem{fan1953minimax}
K.~Fan, ``Minimax theorems,'' \emph{Proceedings of the National Academy of
  Sciences of the United States of America}, vol.~39, no.~1, p.~42, 1953.

\bibitem{matthews2012linear}
W.~Matthews, ``A linear program for the finite block length converse of
  polyanskiy--poor--verd{\'u} via nonsignaling codes,'' \emph{Information
  Theory, IEEE Transactions on}, vol.~58, no.~12, pp. 7036--7044, 2012.

\bibitem{csiszar1998method}
I.~Csisz{\'a}r, ``The method of types [information theory],'' \emph{Information
  Theory, IEEE Transactions on}, vol.~44, no.~6, pp. 2505--2523, 1998.

\bibitem{csiszár2011information}
\BIBentryALTinterwordspacing
I.~Csisz{\'a}r and J.~K{\"o}rner, \emph{Information Theory: Coding Theorems for
  Discrete Memoryless Systems}.\hskip 1em plus 0.5em minus 0.4em\relax
  Cambridge University Press, 2011. [Online]. Available:
  \url{http://books.google.co.il/books?id=2gsLkQlb8JAC}
\BIBentrySTDinterwordspacing

\end{thebibliography}

\end{document}